\documentclass[10pt]{amsart}

\usepackage{amssymb,amsthm,amsmath}
\usepackage[numbers,sort&compress]{natbib}
\usepackage{color}
\usepackage{graphicx}
\usepackage{tikz}

\title[ ]{Continuous quasiperiodic Schr\"odinger  operators with Gordon type potentials}

\author{Wencai Liu}
\address[Wencai Liu]{Department of Mathematics, University of California, Irvine, California 92697-3875, USA} \email{liuwencai1226@gmail.com}

\newcommand{\R}{\mathbb{R}}
\newcommand{\Z}{\mathbb{Z}}

\theoremstyle{plain}
\newtheorem{theorem}{Theorem}[section]
\newtheorem{corollary}[theorem]{Corollary}
\newtheorem{lemma}[theorem]{Lemma}

\theoremstyle{definition}

\newtheorem{remark}[theorem]{Remark}

\begin{document}

 \newcommand{\N}{\mathbb{N}}

\begin{abstract}

Let us concern the quasi-periodic Schr\"odinger operator  in  the continuous case,
\begin{equation*}
    (Hy)(x)=-y^{\prime\prime}(x)+V(x,\omega x)y(x),
\end{equation*}
where $V:(\R/\Z)^2\to \R$ is  piecewisely $\gamma$-H\"older continuous with respect to the second variable.
Let $L(E)$ be the Lyapunov exponent of $Hy=Ey$.
Define $\beta(\omega)$ as
\begin{equation*}
   \beta(\omega)= \limsup_{k\to \infty}\frac{-\ln ||k\omega||}{k}.
\end{equation*}
We prove that $H$ admits no eigenvalue in regime $\{E\in\R:L(E)<\gamma\beta(\omega)\}$.

\end{abstract}
\maketitle

 \section{Introduction}
 In this note, we study the continuous quasi-periodic Schr\"odinger operator, which is given by
 \begin{equation*}
    (Hy)(x)=-y^{\prime\prime}(x)+V(x,\omega x)y(x),
\end{equation*}
where $V:(\R/\Z)^2\to \R$ is   the potential  and $\omega\in\R$ is  frequency.

  We are interested in a particular class of functions $V$.
   $V:(\R/\Z)^2\to \R$ is  called   piecewisely $\gamma$-H\"older continuous with respect to the second variable, denoted by $A_{\gamma}$, if  $V:(\R/\Z)^2\to \R$ is measurable and  there exist $a_1<a_2<\cdots< a_m$ ($a_1=0$ and $a_m=1$) such that
\begin{equation}\label{holder}
 \sup_{(x,y)\in (\R/\Z)^2}|V(x,y)|  + \sup _{x\in \R/\Z}\sup _{1\leq i\leq m-1}\sup_{y_1,y_2\in[a_i,a_{i+1}]\atop y_1\neq y_2}|\frac{V(x,y_1)-V(x,y_2)|}{|y_1-y_2|^{\gamma}}<\infty.
\end{equation}
We emphasize  that   $V\in A_{\gamma}$ implies  $V(x,y)$ is bounded and continuous with
respect to $y$.  $A_{\gamma}$ contains a special case $V=V_1(x)+V_2(\omega x)$, where $V_1:\R/\Z\to \R$  is  bounded measurable function   and $V_2:\R/\Z\to \R$  is  piecewisely $\gamma$-H\"older continuous. We always assume that potentials $V\in A_{\gamma}$  in this paper.

Recently, there has been  a remarkable development of arithmetically spectral   transition (singular spectrum and Anderson localization) for discrete  quasiperiodic operator \cite{marx},  in particular for explicit models: almost Mathieu operator \cite{liu1,liu2,zhou,Ilya}, Maryland model \cite{liu} and extended Harper model \cite{han}(Jacobi operator).
For the continuous, arithmetic phase transitions are currently very far from being established.
The quantitative arguments for discrete case is to show the absence of eigenvalues if
 the frequency can be approximated by a rational number well (Gordon type potential in one dimension) , which improved  the  previous   results\footnote{We should mention that
    Gordon and Nemirovski also obtained   the absence of  eigenvalues for
 discrete   Schr\"odinger operators with Gordon type potentials  in higher dimensions \cite{gordon2017absence}. It is a interesting problem to make this arithmetic  sharp.}   \cite{gordon1976point,simon1982almost}, obtaining sharp thresholds for the smallness of small denominators in terms of the Lyapunov exponents.
 For various other recent developments on the Gordon-type potentials, see Damanik-Stolz \cite{damanik2000generalization},  Damanik's survey paper \cite{damanik} and references therein, and Jitomirskaya-Zhang \cite{zhang}. The purpose of this note is to obtain a similar sharp result for the continuous case.
 We should mention that the   localization part  for continuous case is only  known for  a  full Lebesgue measure subset of Diophantine frequencies \cite{MR3623250}.

 \begin{theorem}\label{thm}
Let $H$ be a quasiperiodic Schr\"odinger operator,
\begin{equation*}
    (Hy)(x)=-y^{\prime\prime}(x)+V(x,\omega x)y(x),
\end{equation*}
where $V:(\R/\Z)^2\to \R$ is piecewisely  $\gamma$-H\"older continuous with respect to the second variable.
Let $L(E)$ be the Lyapunov exponent of $Hy=Ey$.
Define $\beta(\omega)$ as
\begin{equation}\label{equ2}
   \beta(\omega)= \limsup_{k\to \infty}\frac{-\ln ||k\omega||}{k}.
\end{equation}
Then  $H$ does not have any eigenvalue in the regime $\{E\in\R:L(E)<\gamma\beta(\omega)\}$.

\end{theorem}
\begin{remark}
\begin{itemize}
  \item  By the  recent results of \cite{liu,zhou},  the statement of Theorem \ref{thm} for the case $\gamma=1$ is sharp for discrete Schr\"odinger operators with regular and singular potentials. We refer the reader to   \cite{yang} for the potentials with finitely many singular points. We expect that Theorem \ref{thm}  is sharp for $\gamma=1$ and even for $0<\gamma<1$.
  \item  The potential $V$ here is only H\"older continuous. Such Schr\"odinger operator with rough potential
was studied in \cite{mavi,han1}
\end{itemize}

\end{remark}
As a direct corollary, we obtain
\begin{corollary}\label{cor}
Let $H$ be a quasiperiodic Schr\"odinger operator,
\begin{equation*}
    (Hy)(x)=-y^{\prime\prime}(x)+V(x,\omega x)y(x),
\end{equation*}
where $V:(\R/\Z)^2\to \R$ is a H\"older continuous   with respect to the second variable.
Suppose  the frequency $\omega$ satisfies
\begin{equation*}
    \limsup_{k\to \infty}\frac{-\ln ||k\omega||}{k}=\infty.
\end{equation*}
Then $H$ does not have any eigenvalue.

\end{corollary}

If the analytical  potential $V$ is small or the energy $E$ is large, the Lyapunov exponent is zero in the spectrum \cite{Hou}. Thus we have the following result.
\begin{corollary}
Let $H_\lambda$ be a quasiperiodic Schr\"odinger operator,
\begin{equation*}
    (H_\lambda y)(x)=-y^{\prime\prime}(x)+\lambda V(x,\omega x)y(x),
\end{equation*}
where $V:(\R/\Z)^2\to \R$ is  analytic.
Suppose  the frequency $\omega$ satisfies
\begin{equation*}
    \limsup_{k\to \infty}\frac{-\ln ||k\omega||}{k}>0.
\end{equation*}
Then there exists $\lambda_0>0$ such that $H_\lambda$ does not have any eigenvalue for $|\lambda|\leq \lambda_0$.
\end{corollary}
\begin{corollary}
Let $H$ be a quasiperiodic Schr\"odinger operator,
\begin{equation*}
    (H y)(x)=-y^{\prime\prime}(x)+  V(x,\omega x)y(x),
\end{equation*}
where $V:(\R/\Z)^2\to \R$ is  analytic.
Suppose  the frequency $\omega$ satisfies
\begin{equation*}
    \limsup_{k\to \infty}\frac{-\ln ||k\omega||}{k}>0.
\end{equation*}
Then there exists $E_0>0$ such that $H $ does not have any eigenvalue    in regime $[E_0,\infty)$.
\end{corollary}


\section{Proof of Theorem \ref{thm}}
Let $u$ be a solution of $Hu=Eu$.
Define the transfer matrix $T(E,y,x)$ as
\begin{equation}\label{Gtran}
     T(E,y,x)\left[
               \begin{array}{c}
                 u^{\prime}(y) \\
                 u(y) \\
               \end{array}
             \right]=
             \left[
               \begin{array}{c}
                 u^{\prime}(x) \\
                 u(x) \\
               \end{array}
             \right].
\end{equation}
\begin{lemma} \cite{simon1982almost}\label{Lemgordonidea1}
Let $B\in \text{SL}(2,\mathbb{R})$ and  $\varphi$ be a unit vector in $\mathbb{R}^2$,  then $$\max{\{||B^2\varphi||, ||B\varphi||,||B^{-1}\varphi||\}}\geq \frac{1}{4}.$$
\end{lemma}

Let $\frac{p_n}{q_n}$ be the continued fraction expansion to $\omega$.
By \eqref{equ2}, one has
\begin{equation}\label{cont}
    \limsup_{n\to \infty}\frac{\ln q_{n+1}}{q_n}=\beta(\omega).
\end{equation}
Below, $\varepsilon>0$ is   arbitrarily small, and $c$($C$) is small (large) constant depending on $E$, $\omega$ and  the potential $V$.
Since $V\in A_{\gamma}$, there exist  $a_1,a_2,\cdots, a_m$ ($a_1=0$ and $a_m=1$) such that
\eqref{holder} holds.

Write down the eigen-equation $Hy=Ey$ in first order,
\begin{equation}\label{equ3}
     \left[
       \begin{array}{c}
         y ^{\prime}(x)\\
         y (x)\\
       \end{array}
     \right]^{\prime}=\left[
               \begin{array}{cc}
                 0 &V(x,\omega x)-E\\
                 1 & 0 \\
               \end{array}
             \right]\left[
       \begin{array}{c}
         y ^{\prime}(x)\\
         y (x)\\
       \end{array}
     \right].
\end{equation}
By the definition of transfer matrix and Lyapunov exponent (see \cite{furman1997multiplicative}), we have
\begin{equation}\label{equ4}
  || T(E,x,y)||\leq Ce^{(L(E)+\varepsilon)|x-y|}.
\end{equation}
For simplicity, sometimes we ignore the dependence of $E$ in $T(E,x,y)$ and $L(E)$.
%
%
%
By \eqref{cont}, there exists a sequence $\{q_{n_k}\}$ such that
\begin{equation*}
    q_{n_k+1}\geq e^{(\beta(\omega)-\varepsilon)q_{n_k}}.
\end{equation*}
By the property of continued fraction expansion, one has
\begin{equation}\label{small}
    ||q_{n_k}\omega||\leq e^{-(\beta(\omega)-\varepsilon)q_{n_k}},
\end{equation}
where $||x||=\text{dist}(x,\Z)$.

We  start the proof with  some basic Lemmas.
\begin{lemma}\label{lerough}
The following estimate holds,
\begin{equation*}
   \int_{0} ^{q_{n_k}}|V(t,\omega t)-V(t,\omega (t+q_{n_k}))|dt \leq C e^{-(\gamma\beta-\varepsilon) q_{n_k}}.
\end{equation*}
\end{lemma}
\begin{proof}
Let
\begin{equation*}
    I=\{t\in[0,q_{n_k}]: \text{dist}(\omega t,a_i+\ell)\geq 2 e^{-(\beta(\omega)-\varepsilon)q_{n_k}} \text{ for all }i=1,2,\cdots,m-1 \text{ and } \ell=0,1,2,\cdots,q_{n_k}-1\}.
\end{equation*}
By \eqref{small} and the assumption on $V$, one has
\begin{equation*}
    |V(t,\omega t)-V(t,\omega (t+q_{n_k}))|dt \leq  C e^{-(\gamma\beta-\varepsilon) q_{n_k}} \text{ for } t\in I.
\end{equation*}
We also have
\begin{equation*}
    |I^c|\leq Cq_{n_k} e^{-(\beta(\omega)-\varepsilon)q_{n_k}}\leq Ce^{-(\beta(\omega)-\varepsilon)q_{n_k}},
\end{equation*}
where $I^c=[0,q_{n_k}]\backslash I$.
Thus we obtain
\begin{eqnarray*}
   \int_{0} ^{q_{n_k}}|V(t,\omega t)-V(t,\omega (t+q_{n_k}))|dt &=&  \int_{I}|V(t,\omega t)-V(t,\omega (t+q_{n_k}))|dt+ \int_{I^c}|V(t,\omega t)-V(t,\omega (t+q_{n_k}))|dt \\
   &\leq& C e^{-(\gamma\beta-\varepsilon) q_{n_k}} .
\end{eqnarray*}
\end{proof}
\begin{lemma}
The following estimates hold
\begin{equation}\label{Gro1}
    ||T(E,0,q_{n_k})-T(E,q_{n_k},2q_{n_k})||\leq Ce^{(L(E)-\gamma\beta(\omega)+\varepsilon)q_{n_k}},
\end{equation}
and
\begin{equation}\label{Gro2}
    ||T(E,0,-q_{n_k})-T(E,0,q_{n_k})^{-1}||\leq Ce^{(L(E)-\gamma\beta(\omega)+\varepsilon)q_{n_k}}.
\end{equation}
\end{lemma}
\begin{proof}
Consider the two differential equations $y_1,y_2$ on $[0,q_{n_k}]$ with the same initial condition at $0$,
\begin{equation*}
      \left[
       \begin{array}{c}
         y _1^{\prime}(x)\\
         y_1 (x)\\
       \end{array}
     \right]^{\prime}=\left[
               \begin{array}{cc}
                 0 &V(x,\omega x)-E\\
                 1 & 0 \\
               \end{array}
             \right]\left[
       \begin{array}{c}
         y_1 ^{\prime}(x)\\
         y_1 (x)\\
       \end{array}
     \right]
\end{equation*}
and
\begin{equation*}
  \left[
       \begin{array}{c}
         y_2 ^{\prime}(x)\\
         y _2(x)\\
       \end{array}
     \right]^{\prime}=\left[
               \begin{array}{cc}
                 0 &V(x+q_{n_k},\omega (x+q_{n_k}))-E\\
                 1 & 0 \\
               \end{array}
             \right]\left[
       \begin{array}{c}
         y_2 ^{\prime}(x)\\
         y _2(x)\\
       \end{array}
     \right].
\end{equation*}
Let $Y(x)=\left[
            \begin{array}{c}
              y_1^{\prime}(x)-y_2^{\prime}(x) \\
              y_1(x)-y_2(x) \\
            \end{array}
          \right]
$. 
Denote  by
\begin{eqnarray*}
   F(t) &=& \left[
                   \begin{array}{cc}
                     0 & V(t,\omega t)-V(t,\omega (t+q_{n_k})) \\
                     0& 0\\
                   \end{array}
                 \right]T(E,0,t)\left[
                   \begin{array}{c}
                     y_2^\prime(0) \\
                     y_2(0) \\
                   \end{array}
                 \right] \\
   &=& \left[
                   \begin{array}{cc}
                     0 & V(t,\omega t)-V(t,\omega (t+q_{n_k})) \\
                     0& 0\\
                   \end{array}
                 \right]\left[
                          \begin{array}{c}
                            y_2 ^\prime(t)\\
                          y_2(t) \\
                          \end{array}
                        \right],
\end{eqnarray*}
and
\begin{equation*}
    T(t)=T(E,0,t).
\end{equation*}
Thus
\begin{equation}\label{Gmar15}
    Y^\prime(x)=F(x)+\left[
               \begin{array}{cc}
                 0 &V(x,\omega x)-E\\
                 1 & 0 \\
               \end{array}
             \right]Y(x),
\end{equation}
and $Y(0)=0$,  since $V(x+q_{n_k},\omega (x+q_{n_k}))=V(x,\omega (x+q_{n_k}))$.
By  the constant variation method, we obtain that
\begin{eqnarray}
  Y(x) &=&  T(x)\int_0^xT^{-1}(t) F(t)dt\nonumber\\
   &=&\int_0^x  T(x)T^{-1}(t) F(t)dt.\label{small2}
\end{eqnarray}
We give the proof of \eqref{small2} in the Appendix.

By   Lemma \ref{lerough},
we have for $0<x\leq q_{n_k}$,
\begin{eqnarray}
 \int_{0} ^{x}||F(t) || dt&\leq& || \left[
                          \begin{array}{c}
                            y_2 ^\prime(t)\\
                          y_2(t) \\
                          \end{array}
                        \right]||  \int_{0}^x ||\left[
                   \begin{array}{cc}
                     0 & V(t,\omega t)-V(t,\omega (t+q_{n_k})) \\
                     0& 0\\
                   \end{array}
                 \right] ||dt\nonumber \\
    &\leq&  e^{-(\gamma\beta-\varepsilon) q_{n_k}}||T(E,0,t)||\;\;|| \left[
                          \begin{array}{c}
                            y_2 ^\prime(0)\\
                          y_2(0) \\
                          \end{array}
                        \right]||\nonumber\\
    &\leq&  Ce^{-(\gamma\beta-\varepsilon) q_{n_k}}e^{(L+\varepsilon)t}|| \left[
                          \begin{array}{c}
                            y_2 ^\prime(0)\\
                          y_2(0) \\
                          \end{array}
                        \right]||,\label{small3}
\end{eqnarray}
where the second inequality holds by \eqref{equ4}.

For $x>t$, by \eqref{equ4} again,
one has
\begin{eqnarray}\label{small1}
  ||T(x)T^{-1}(t)|| &=& ||T(E,t,x)||\nonumber \\
  &\leq& Ce^{(L+\varepsilon)|x- t|}.
\end{eqnarray}
By \eqref{small2}, \eqref{small3} and \eqref{small1}, we obtain
\begin{equation*}
    || Y(q_{n_k})||\leq  Ce^{(L-\gamma\beta +\varepsilon)q_{n_k}}|| \left[
                          \begin{array}{c}
                            y_2 ^\prime(0)\\
                          y_2(0) \\
                          \end{array}
                        \right]||,
\end{equation*}
which implies \eqref{Gro1} by the arbitrary choice of $ \left[
                          \begin{array}{c}
                            y_2 ^\prime(0)\\
                          y_2(0) \\
                          \end{array}
                        \right]. $
%

By the fact that
$T(E,0,q_{n_k})^{-1}=T(E,q_{n_k},0) $ and following the similar arguments of proof of \eqref{Gro1},
we can prove \eqref{Gro2}.
\end{proof}
The following lemma is well known for Schr\"odinger operator. We give the proof for completeness.
\begin{lemma}\label{key2}
Suppose $u $ is an eigensolution, that is $Hu=Eu$ and $u\in L^2(\R)$ for some $E$, then
\begin{equation}\label{Gsup}
    \limsup_{x\to\infty}|| \left[
                            \begin{array}{c}
                              u^\prime(x)\\
                              u(x) \\
                            \end{array}
                          \right]
    ||=0.
\end{equation}

\end{lemma}
\begin{proof}
Suppose
\begin{equation*}
   || \left[
                            \begin{array}{c}
                              u^\prime(x_0)\\
                              u(x_0) \\
                            \end{array}
                          \right]
    ||\geq \epsilon.
\end{equation*}
By equation \eqref{equ3},  one has
\begin{equation*}
  \int_{x_0-1}^{x_0+1}  (|u^\prime(x)|+|u (x)|)dx\geq  c\epsilon.
\end{equation*}
By a result of \cite[Lemma 3.1]{simonb}, we have
\begin{equation*}
    \int_{x_0-1}^{x_0+1}|u^{\prime}(x)|dx\leq  C\int_{x_0-2}^{x_0+2}|u (x)|dx.
\end{equation*}
Thus
\begin{equation*}
  \int_{x_0-2}^{x_0+2} | u(x)|dx\geq  c\epsilon,
\end{equation*}
which leads to \eqref{Gsup} by the fact that $u\in L^2(\R)$.
\end{proof}
{\bf Proof of Theorem \ref{thm}}
\begin{proof}
Let $E$ be such that $L(E)<\gamma\beta(\omega)$.
Suppose $Hy=Ey$ and $y\in L^2(\R)$.
Setting $B=T(E,0,q_{n_k})$ and applying Lemma \ref{Lemgordonidea1},
we have
 $$\max{\{||B^2\varphi||, ||B\varphi||,||B^{-1}\varphi||\}}\geq \frac{1}{4},$$
 where
 $\varphi=\left[
            \begin{array}{c}
              y^\prime(0) \\
              y(0) \\
            \end{array}
          \right]
 $  is unit.
 We claim that
 \begin{equation*}
   \max{\{|| \left[
                            \begin{array}{c}
                              y^\prime(-q_{n_k})\\
                              y(-q_{n_k}) \\
                            \end{array}
                          \right]
    ||,|| \left[
                            \begin{array}{c}
                              y^\prime(q_{n_k})\\
                              y(q_{n_k}) \\
                            \end{array}
                          \right]
    ||,|| \left[
                            \begin{array}{c}
                              y^\prime(2q_{n_k})\\
                              y(2q_{n_k}) \\
                            \end{array}
                          \right]
    ||\}}\geq \frac{1}{8}.
 \end{equation*}
 If $||\left[
                            \begin{array}{c}
                              y^\prime(q_{n_k})\\
                              y(q_{n_k}) \\
                            \end{array}
                          \right]||=||B\varphi||\geq \frac{1}{4}$, there is nothing to prove.
 If $||B^{-1}\varphi||\geq \frac{1}{4}$, by \eqref{Gro2},
 \begin{equation*}
    ||\left[
                            \begin{array}{c}
                              y^\prime(-q_{n_k})\\
                              y(-q_{n_k}) \\
                            \end{array}
                          \right]||=||T(E,0,-q_{n_k})\varphi||\geq \frac{1}{8}.
 \end{equation*}
 If $||B^{2}\varphi||\geq \frac{1}{4}$ and $||B \varphi||< \frac{1}{4} $, by \eqref{Gro1},
 \begin{eqnarray*}
    ||\left[
                            \begin{array}{c}
                              y^\prime(2q_{n_k})\\
                              y(2q_{n_k}) \\
                            \end{array}
                          \right]-B^2\varphi|| &=& ||T(E,q_{n_k},2q_{n_k})B\varphi-B^2\varphi|| \\
    &\leq &||T(E,q_{n_k},2q_{n_k})-B|| \;\;\;||B\varphi|| \\
    &\leq& e^{(L(E)-\gamma\beta(\omega)+\varepsilon)q_{n_k}}\leq \frac{1}{8}.
 \end{eqnarray*}
 This implies
 \begin{equation*}
    ||\left[
                            \begin{array}{c}
                              y^\prime(2q_{n_k})\\
                              y(2q_{n_k}) \\
                            \end{array}
                          \right]||\geq \frac{1}{8}.
 \end{equation*}
 We finish the proof of the claim. By Lemma \ref{key2}, this is impossible.

\end{proof}
\section{Appendix. Proof of \eqref{small2}}
\begin{proof}
Let $\widetilde{Y}$ be
\begin{eqnarray}
 \widetilde{Y}(x) &=&  T(x)\int_0^xT^{-1}(t) F(t)dt\nonumber\\
   &=&\int_0^x  T(x)T^{-1}(t) F(t)dt.\label{small2mar15}
\end{eqnarray}
Obviously,
\begin{equation*}
  \widetilde{Y}(0)=0.
\end{equation*}
By the definition of $T(x)$ and \eqref{Gtran}, we have
\begin{equation}\label{Gmar15app}
  T^{\prime}(x)=\left[
               \begin{array}{cc}
                 0 &V(x,\omega x)-E\\
                 1 & 0 \\
               \end{array}
             \right]T(x).
\end{equation}
By \eqref{small2mar15} and \eqref{Gmar15app},
 we have
 \begin{equation*}
    \widetilde{Y}^\prime(x)=F(x)+\left[
               \begin{array}{cc}
                 0 &V(x,\omega x)-E\\
                 1 & 0 \\
               \end{array}
             \right]\widetilde{Y}(x).
\end{equation*}
Thus $Y$ and $ \widetilde{Y}$ satisfy the same differential equation and initial condition. This implies $Y\equiv \widetilde{Y}$.
\end{proof}
 \section*{Acknowledgments}
I would like to thank Svetlana Jitomirskaya for comments on earlier versions of the manuscript and   the  anonymous referee for careful reading of
the manuscript that has led to an important improvement.
 This research was   supported   by  the AMS-Simons Travel Grant (2016-2018), NSF DMS-1401204 and NSF  DMS-1700314.


\end{document}